\documentclass[11pt, reqno]{amsart}

\usepackage{amsmath}
\usepackage{amssymb}
\usepackage{amsthm}
\usepackage{amscd, amsfonts}
\usepackage[dvips]{epsfig}
\usepackage{verbatim}
\usepackage{epsf}
\usepackage[bookmarksnumbered,pdfpagelabels=true,plainpages=false,colorlinks=true,
            linkcolor=black,citecolor=black,urlcolor=blue]{hyperref}
\usepackage{cases}

\theoremstyle{plain}
\newtheorem{theorem}{Theorem}[section]
\newtheorem{lemma}[theorem]{Lemma}

\newtheorem{coro}[theorem]{Corollary}

\theoremstyle{definition}
\newtheorem{rema}[theorem]{Remark}
\newtheorem{defi}[theorem]{Definition}
\newtheorem{notation}[theorem]{Notation}
\newtheorem{assump}[theorem]{Assumption}

\numberwithin{equation}{section}

\newcommand{\rar}{\rightarrow}

\newcommand{\cF}{\mathcal F}

\newcommand{\cS}{\mathcal S}

\newcommand{\al}{\alpha}
\newcommand{\be}{\beta}

\newcommand{\de}{\delta}

\newcommand{\RR}{\mathbb R}

\newcommand{\ve}{\varepsilon}
\newcommand{\Vol}{\operatorname{Vol}}

\newcommand{\p}{\parallel}

\title[Positive energy vacua]{Remarks on positive energy vacua via effective  potentials in string theory}

\author[Dabholkar]{Sujan P. Dabholkar}
\address{C. N. Yang Institute for Theoretical Physics,
Stony Brook University, Stony Brook, NY, USA, 11794}
\email{sujan.dabholkar@stonybrook.edu}
 \thanks{We would like to thank Michael R. Douglas for fruitful discussions on this topic at various stages.}

\author[Disconzi]{Marcelo M. Disconzi}
\address{Department of Mathematics,
Vanderbilt University, Nashville, TN, USA, 37240}
\email{marcelo.disconzi@vanderbilt.edu}

\author[Pingali]{Vamsi P. Pingali}
\address{Department of Mathematics, 
John Hopkins, Baltimore, MD, USA, 21218}
\email{vpingal1@jhu.edu}

\begin{document}

\begin{abstract}
We study warped compactifications of string/M theory with the help of effective potentials, continuing previous work of the last two authors
and Michael R. Douglas presented in
\cite{Disconzi:2012a}. The dynamics of the conformal factor of the internal metric, which
 is responsible for instabilities in these constructions, is explored, and such instabilities 
are investigated in the context of de Sitter vacua. We prove existence results for the 
equations of motion in the case
of a slowly varying warp factor, and the stability of such solutions is also addressed.
These solutions are a
family of meta-stable de Sitter vacua from type IIB string theory in a general 
non-supersymmetric setup.
\end{abstract}

\maketitle

\tableofcontents

\section{Introduction}
It is well known that consistency of String Theory requires a $10$-dimensional space-time,
while maximal supergravity and its quantum version 
called ``M theory'' make sense in $11$-dimensional space-time.
In both cases, one  makes contact with standard $4$-dimensional physics by compactifying the 
extra dimensions to a small $n=6$ or $7$-dimensional compact manifold $M$ ---
obtaining in this fashion
a lower dimensional quantum theory of gravity with matter.

A primary goal of the work on compactifications is to derive an effective action in $4$ 
dimensions --- i.e., an action that could reproduce the observed $4$-dimensional physics.
This effective action is a functional of the $4$-dimensional metric and whatever additional data parametrize the extra
dimensions --- its metric, and the other fields of supergravity or superstring theory
 --- taken as functions on
$4$-dimensional space-time.  Critical points of this effective action, in the usual sense of a variational principle,
correspond to critical points of the original higher-dimensional supergravity or superstring action. 

Earlier works on compactifications have relied heavily 
on supersymmetry, and supersymmetric 
constraints have been used to understand several aspects of  
string compactifications. 
The wide physical understanding brought by the study of supersymmetric 
models notwithstanding, there are at least two good reasons for investigating
effective potentials that do not incorporate supersymmetry.
The first is that, if it exists, supersymmetry is an exact symmetry of nature
only at energy scales far beyond the validity of many of the effective descriptions.
The second reason is the strong  evidence that the cosmological constant, or 
vacuum energy, of our universe is positive. In the simplest effective 
descriptions of string theory, the vacuum energy of  $4$-dimensional space-time
is given by an effective potential $V_{eff}$. Persistent physical features, like the sign of
 the cosmological
constant, should typically be described by meta-stable local minima of $V_{eff}$. 
However, effective potentials with local minima corresponding to positive
vacuum energy do not in general allow supersymmetry.

Here we shall be concerned with what can be called cosmological constraints
for de-Sitter (dS) vacua. In other words, we consider the case
of a maximally symmetric $4$-dimensional space-time and seek conditions that guarantee
the existence of meta-stable positive local minima of $V_{eff}$.
Our focus will be on compactifications with Dq-branes and/or  Oq-planes
and Type IIB strings.

Many authors contributed to our current understanding
of effective descriptions in string theory, and 
a thorough review would be beyond the scope of this manuscript.
A detailed and seminal discussion of the matter can be found in
\cite{Douglas:2009a}, with subsequent properties investigated in
\cite{Disconzi:2012a}. The interested reader should also consult 
\cite{AGMP, CHSW, Douglas:2006ab, GKP, Grana:2006ab,
KKLT,  EvaS} and references therein for further details.

\section{Setting and the basic equations\label{basic_section}}

Consider compactification on an $n = D - d$-dimensional compact manifold $M$ to
$d$-dimensional maximally symmetric space-time (Minkowski, AdS, dS).
In the $D$-dimensional space, 
consider General Relativity coupled
to matter, the latter being encoded as usual in a set of field strengths $F^{(p)}$,
$p=1,\dots, L$ (these are curvature terms, with the standard curvature of the Yang-Mills functional being
the canonical example). The full $D$-dimensional metric is assumed 
to have the form of a Kaluza-Klein warped metric with a conformal factor,
\begin{equation}
ds^{2}= e^{2A(x)}\eta_{\mu \nu}dz^{\mu}dz^{\nu} + e^{2B(x)}g_{ij}(x)dx^{i}dx^{j},
\nonumber
\end{equation}
where $\eta_{\mu \nu}$ is a metric on the $4$-dimensional space-time (Minkowski, dS, AdS) with  $z^\mu$ coordinates on it, $g_{ij}$ is a 
metric on the internal manifold $M$, $x^i$ are coordinates on $M$, and $x\in M$.

\begin{notation}
We shall adopt the notations and conventions of
\cite{Disconzi:2012a}.
\end{notation}

\begin{assump} From now on we assume that $d=4$ and $n=6$.
For simplicity, all quantities involved are assumed to be smooth\footnote{Orbifolds could 
also be included, in which case quantities should be smooth away
from singularities. We have not treated orbifolds here
to avoid technicalities; they will be the focus of a future work
\cite{in_prep}.} unless stated 
differently. Since many of the fields involved are usually distributional quantities, 
our point of view is that they have been properly smeared by, for example,
convoluting against smooth functions.
\label{assumption_smooth}
\end{assump}

\begin{rema}
The above smoothness assumption can certainly be relaxed with 
no difficulties. In fact, our existence theorems will hold in Sobolev spaces, 
so it suffices to assume that our fields have only a finite number of derivatives. 
A possible exception is 
the ``string term" $T_{string}$ (see below). Such a term is, in general,
a distribution supported on hyperplanes. Hence, whenever necessary, it will
be assumed that fields have been properly smeared or smoothed out,
as mentioned above.
The smearing of
$T_{string}$ notwithstanding, we point out that many of our results will remain
valid, if however only suitable ``integral bounds" --- which allow for 
distributional coefficients --- are imposed on $T_{string}$
similarly to what was done in \cite{Disconzi:2012a}.
\label{remark_smooth}
\end{rema}

Following the construction of $V_{eff}$ as in \cite{Douglas:2009a} yields
\begin{align}
\begin{split}
V_{eff}&= \frac{1}{2}\int_{M} \big ( -u^{2}v^{2}R-5\nabla v \cdot \nabla(u^{2}v)
-3v^{2} |\nabla u |^{2}+\frac{u^2}{2}\sum_{p=0}^L v^{(3-p)}|F_{p}|^2  \\
&-u^{2}v^{(q-3)/2}T_{string} \big )+\alpha \Big (\frac{1}{G_N}-\int_{M} uv^3 \Big )+
\beta \Big (\Vol(M)-\int_{M} v^3 \Big),
\end{split}
\nonumber
\end{align}
where $u=e^{2A}$ is called the warp factor, $v=e^{2B}$ is called the conformal factor, 
$R$ is the scalar curvature of $g$, 
 $\al$ and $\be$
are constants\footnote{These are in fact Lagrange multipliers in the sense 
that their variational equations enforce constraints.}, $G_N$ is Newton's constant, $p$, $q$ and $L$ are integers that
depend on the particular model under consideration, and integrals are
with respect to the natural volume element given by $g$. The dot 
``$\cdot$" is the inner product on the metric $g$, but we shall omit it when no
confusion can arise and write simply $\nabla u  \nabla v$ etc.
The term $T_{string}$ is added ad hoc to incorporate the non-classical contributions to the effective 
potential coming from string/M-theory. 

It is shown in \cite{Douglas:2009a} that once the warped constraint 
is imposed (see below), the Lagrange multiplier $\al$ becomes the 4-d space-time scalar curvature. 
This provides a setup for justifying many effective potentials that have been studied in 
the
context of string/M theory compactifications, especially regarding dS solutions. 
In particular, a given vacuum corresponds to a dS solution, and thus has a
positive vacuum energy, if and only if $\al > 0$. 
The other Lagrange multiplier, $\be$, is used to obtain the minimum for 
the volume modulus related to the conformal factor. Also, it is important to notice that arguments related to Chern-Simons terms and warping done in \cite{Douglas:2009a} will also hold here, as those terms
do not depend on the conformal factor.

We  consider compactifications with Dq-branes or/and Oq-planes, and we
want to 
study critical points of $V_{eff}$. 
The first variation of $V_{eff}$ with respect to $u$ and $v$ in the direction 
of $\psi$ and $\varphi$  are, respectively, 
\begin{align}
\begin{split}
\frac{\delta V_{eff}}{\delta u } (\psi) = & \, \frac{1}{2}\int_{M} \big ( -2uv^{2}R
+10uv \Delta   v +6\nabla(v^{2}\nabla u)+u\sum_{p=0}^Lv^{(3-p)}|F_{p}|^2 \\
&-2uv^{(q-3)/2}T_{string} \big )\psi -\alpha\int_{M} v^3 \psi,
\end{split}
\nonumber
\end{align}
and
\begin{align}
\begin{split}
\frac{\delta V_{eff}}{\delta v } (\varphi)
 = & \, \frac{1}{2}\int_{M}\big ( -2u^{2}vR+5\Delta (u^{2}v)+5u^{2}\Delta  v 
 -6v|\nabla u|^{2}  \\
&+\frac{u^2}{2}\sum_{p=0}^L(3-p)v^{(2-p)}|F_{p}|^2-\frac{(q-3)}{2}u^2v^{(q-5)/2}T_{string} \big )
\varphi  \\
&  -3 \alpha\int_{M}  uv^{2} \varphi
- 3 \beta\int_{M} v^{2} \varphi .
\end{split} \nonumber
\end{align}
We must satisfy $\frac{\delta V_{eff}}{\delta u}= \frac{\delta V_{eff}}{\delta v}= 0$ at 
each critical point. From the above we obtain the following equations of motion:
\begin{subnumcases}{\label{eq_motion}}
10uv \Delta v+6\nabla(v^{2}\nabla u) -2uv^{2}R +u\sum_{p=0}^Lv^{(3-p)}|F_{p}|^2-2uv^{(q-3)/2}T_{string}=2\alpha v^3, & \label{eq_motion_1}  \\
5\Delta(u^{2}v)+5u^{2}\Delta v -2u^{2}vR -6v|\nabla u|^{2}+\frac{u^2}{2}
\sum_{p=0}^L(3-p)v^{(2-p)}|F_{p}|^2 &  \nonumber \\ 
\hspace{1cm} -\frac{(q-3)}{2}u^{2}v^{(q-5)/2}T_{string} 
 =6[\alpha uv^2+\beta v^2]. &  \label{eq_motion_2}
\end{subnumcases}
These are subject to the constraints
\begin{subnumcases}{\label{constraints}}
\int_{M} uv^3 = \frac{1}{G_N},   & \label{constraint_1}  \\
\int_{M} v^3 = \Vol(M). &  \label{constraint_2}
\end{subnumcases}
Equation (\ref{constraint_1}) is sometimes referred to as the warped constraint.

The second variation of $V_{eff}$ with respect to $v$ and in the direction 
$\varphi$ is
\begin{align}
\begin{split}
\frac{\de^2 V_{eff}}{\de v^2}(\varphi) = &   \, 
\frac{1}{2}\int_{M} \big ( -2u^{2}R
-6 |\nabla u |^{2}
+\frac{u^2}{2}\sum_{p=0}^L(3-p)(2-p)v^{(1-p)}|F_{p}|^2  \\ &
-\frac{(q-3)(q-5)}{4}u^2v^{(q-7)/2}T_{string} \big ) \varphi^2 
 - 5 \int_M \nabla \varphi \nabla(u^2 \varphi)   \\
& -6  \alpha \int_{M} uv \varphi^2
- 6 \beta\int_{M} v \varphi^2 .
\end{split}
\nonumber
\end{align}
\begin{defi}
The  \emph{mass squared of the conformal factor} or \emph{volume modulus}, 
denoted
$\frac{\partial^2 V_{eff}}{\partial v^2}$,
  is defined by
  \begin{gather}
\frac{\partial^2 V_{eff}}{\partial v^2} : =\left. \frac{\de^2 V_{eff}}{\de v^2}(\varphi) \right|_{\varphi = 1} .
\nonumber
\end{gather}
Solutions $(v,u)$ of (\ref{eq_motion}) such that 
\begin{gather}
\left. \frac{\partial^2 V_{eff}}{\partial v^2} \right|_{(v,u)} > 0
\nonumber
\end{gather}
are called stable\footnote{Here we use a slight abuse of language, as such 
a condition would be better called
meta-stable, since tunneling to other vacua can occur. We shall, however, use 
the terms
stable and meta-stable interchangeably.} and unstable otherwise.
\label{defI_modulus}
\end{defi}
In our case,
\begin{align}
\begin{split}
\frac{\partial^2 V_{eff}}{\partial v^2} = &   \, 
\frac{1}{2}\int_{M} \big ( -2u^{2}R
-6 |\nabla u |^{2}
+\frac{u^2}{2}\sum_{p=0}^L(3-p)(2-p)v^{(1-p)}|F_{p}|^2  \\ &
-\frac{(q-3)(q-5)}{4}u^2v^{(q-7)/2}T_{string} \big ) 
  -6  \alpha \int_{M} uv 
- 6 \beta\int_{M} v .
\end{split}
\label{mass_general}
\end{align}
We are interested in understanding the stability associated with the conformal 
factor $v$. Mostly in supersymmetric solutions, the warp factor $u$ is 
related to $v$, and such solutions are stable. It is known, however, that
one has to deal with instabilities in de Sitter vacua obtained from string compactications. 
The KK mode mostly responsible for such instabilities is the conformal factor of the
metric. Thus, a reasonable strategy is to hold the other fields coming from 
compactifications fixed and study  the dynamics of the
 conformal factor $v$ along with the 
warp factor $u$. This can be done for general string/supergravity compactifications, but 
we mainly study theories related to Type IIB strings in this paper. In \cite{KKLT}, it 
was found that in Type IIB string theory, one can achieve stability by fixing all massless 
fields to AdS vacuum with the help of non-perturbative effects. The authors add to the 
potential a term like an Anti-D3 brane. For suitable choices of the 
added potential term, the
AdS minimum becomes a dS minimum, but the rest of the potential does not change 
significantly (this is the main idea behind the uplifting construction). 
This minimum is meta-stable. 
It is 
unstable to either quantum tunneling or thermal excitations over a barrier, 
in which case the Universe goes to infinity in moduli space after some time. Since we are 
not dealing with supersymmetric setups, we would like to stabilize $v$ by finding 
conditions that imply $\frac{\partial^{2}V_{eff}}{\partial v^2} > 0$ in general, and
we suggest 
that non-perturbative effects are very small to disturb the minimum. Notice that as we 
are primarily interested in dS space, it will be natural to consider $\alpha > 0$ in much 
of what follows below. 

\section{Slowly varying warp factor and (in)stability analysis}
One commonly investigated case is that of a slowly varying warp factor, i.e.,
$\nabla u \approx 0$ (see e.g. \cite{Douglas:2009a} and references therein).
Here we consider two situations where it is shown that the system (\ref{eq_motion})
can be solved for $u$ sufficiently close to a constant.
In one case, we shall obtain instability of the volume modulus, whereas 
in the second case, stability will be demonstrated.
Our methods are based on the implicit function theorem and
they also involve a perturbation of the coefficients of the equations.
We comment on the legitimacy of this perturbation at the end.
We do not
necessarily impose (\ref{constraints}) at this point, and a more thorough investigation
of the existence of solutions to (\ref{eq_motion}) will be carried out
in a future work \cite{in_prep}.

\subsection{Unstable solutions}
Assume $\al > 0$, and 
consider first the case of a constant warp factor. Plugging  $u=constant$ 
in (\ref{eq_motion}), we see that upon the redefinition 
$\al \mapsto \al/u$ and $\be \mapsto \be/u^2$ we can assume $u \equiv 1$.
Setting $u=1$ implies that both  
(\ref{eq_motion_1}) and (\ref{eq_motion_2}) hold if 
\begin{gather}
F_p = 0 \text{ except for } p =1,\, q = 7, \text{ and } 
\be = -\frac{2\al}{3},
\nonumber
\end{gather}
and we henceforth suppose so, in which case both equations reduce to
\begin{gather}
10 \Delta v - 2R v + |F_1|^2 v - 2 T_{string} v -2\al v^2 = 0,
\label{eq_v_sub_sup}
\end{gather}
provided that $v > 0$.
Equation (\ref{eq_v_sub_sup}) can be solved by the method of sub- and
super-solutions (see e.g. \cite{SY} or \cite{DK} for the case with boundary).
Write the equation as
\begin{gather}
 \Delta v + f(v) = 0.
\nonumber
\end{gather}
We seek functions $v_-$ and $v_+$ such that $v_- \leq v_+$,
\begin{gather}
 \Delta v_- + f(v_-) \geq 0,
\nonumber
\end{gather}
and
\begin{gather}
 \Delta v_+ + f(v_+) \leq  0.
\nonumber
\end{gather}
Let $v_- \equiv constant > 0$. The differential inequality for $v_-$ then becomes
\begin{gather}
v_- \leq \frac{ \frac{1}{2} |F_1|^2 - T_{string} - R }{\al}.
\nonumber
\end{gather}
Hence, if 
\begin{gather}
\frac{1}{2} |F_1|^2 - T_{string} - R > 0,
\nonumber
\end{gather}
 we can choose
$v_-$ so small that the above inequality is satisfied. Similarly, the differential
inequality for $v_+ \equiv constant$ becomes
\begin{gather}
v_+ \geq \frac{ \frac{1}{2} |F_1|^2 - T_{string} - R }{\al},
\nonumber
\end{gather}
which will be satisfied by choosing $v_+$ sufficiently large.
The method of sub- and super-solutions now implies the existence of a
smooth solution $v_*$ to (\ref{eq_v_sub_sup}). This solution is positive because it satisfies $v_- \leq v_* \leq v_+$.

Solutions in the neighborhood of $v=v_*$, $u = 1$ can now be obtained
with the help of implicit function-type theorems.
Consider 
\begin{gather}
M(v,u) = 
10 uv \Delta v + 6v^2 \Delta u + 12 v \nabla v \cdot \nabla u
-2u v^2 R + u \sum_{p=0}^L v^{3-p} |F_p|^2
- 2u v^2 T_{string} - 2\al v^3,
\nonumber 
\end{gather}
and
\begin{gather}
N(v,u) =
10 u^2 \Delta v + 10 uv \Delta u + 20 u \nabla u \cdot \nabla v + 10
uv |\nabla u|^2 -2u^2 v R + 4v |\nabla u |^2
\nonumber \\
 + \frac{u^2}{2} \sum_{p=0}^L (3-p)v^{2-p} |F_p|^2 
 -2 u^2 v T_{string} - 6\al u v^2 - 6\be v^2.
\nonumber
\end{gather}
Let 
\begin{gather}
h = -2R + |F_1|^2 - 2T_{string}.
\nonumber
\end{gather}
A solution to (\ref{eq_motion}) with $q=7$
is then given by $M(v,u) = 0 = N(v,u)$. As we are interested in positive solutions,
we can factor $v$ from $M(v,u)$ and look equivalently for solutions
of $\widetilde{M}(v,u) = 0 = N(v,u)$, where 
\begin{gather}
\widetilde{M}(v,u) = 
10 u \Delta v + 6v \Delta u + 12  \nabla v \cdot \nabla u
-2u v R + u \sum_{p=0}^L v^{2-p} |F_p|^2
- 2u v T_{string} - 2\al v^2.
\nonumber 
\end{gather}
Let $H^{s} = H^s(M)$ be
the standard Sobolev spaces, where $s$ is large. 
Define a map
\begin{align}
& G: \RR \times \RR \times  H^s \times \ldots\times  H^s \times H^s \rar H^{s-2} \times H^{s-2},
\nonumber \\
& (\al, \be, F_0, F_1, \ldots,R, T_{string}, v,u) \mapsto (\widetilde{M}(v,u), N(v,u)).
\nonumber
\end{align}
We claim that 
$D_{v,u} G (0,-\frac{2 \al}{3}, 0, \dots, 0, v_*,1)(\chi_1 , \chi_2)$ is an isomorphism, 
where $D_{v,u} $ is the derivative with respect
to the last two components. Computing,
\begin{gather}
D_{v,u} G (0,-\frac{2 \al}{3},0,\dots,0,v_*,1)(\chi_1 , \chi_2)  = (P,Q), \nonumber 
\end{gather}
where 
\begin{align}
\begin{split}
P =10 \Delta \chi_1 + 6 v_* \Delta \chi_2 + (h-4\al v_*) \chi_1 + 
12 \nabla v_* \cdot \nabla \chi_2 + 2\al v_*^2 \chi_2,
\end{split}
\nonumber
\end{align}
and
\begin{gather}
Q =  10 \Delta \chi_1 + 10 v_* \Delta \chi_2 + (h-4\al v_*) \chi_1 
+ 20 \nabla v_* \cdot \nabla \chi_2 - 2\al v_*^2 \chi_2 .
\nonumber
\end{gather}
Given $\psi_1,\psi_2 \in H^{s-2}$, we wish to solve
\begin{subnumcases}{\label{linearized}}
 10 \Delta \chi_1 + 6 v_* \Delta \chi_2 + (h-4\al v_*) \chi_1 + 
12 \nabla v_* \cdot \nabla \chi_2 + 2\al v_*^2 \chi_2 = \psi_1,
 & \label{linearized_1}  \\
 10 \Delta \chi_1 + 10 v_* \Delta \chi_2 + (h-4\al v_*) \chi_1 
+ 20 \nabla v_* \cdot \nabla \chi_2 - 2\al v_*^2 \chi_2 = \psi_2. &  \label{linearized_2}
\end{subnumcases}
Subtracting (\ref{linearized_1}) from (\ref{linearized_2}) produces
\begin{gather}
4 v_* \Delta \chi_2  + 8 \nabla v_* \cdot \nabla \chi_2 - 4\al v_*^2 \chi_2 = \psi_2
- \psi_1.
\label{decoupled}
\end{gather}
Since $\al > 0$ and $v_* > 0$, a standard argument with the maximum principle
and the Fredholm alternative 
shows that equation (\ref{decoupled}) has a unique $H^s$ solution.
Plugging it back into (\ref{linearized_1}) gives a scalar equation for $\chi_1$ 
only, which will, again by the  maximum principle and the Fredholm alternative, have a unique solution provided that
\begin{gather}
h - 4\al v_* < 0.
\nonumber
\end{gather}
We conclude that if the above is satisfied, then
$D_{v,u} G (0,-\frac{2 \al}{3}, 0, \dots, 0, v_*,1)$ is an 
isomorphism.
 Invoking now the implicit function theorem and recalling that $v_- \leq v_* \leq v_+$, we conclude the following.
\begin{theorem}
Let $\al >0$, $q=7$, and fix a 
sufficiently large real number $s$.
Denote by $v_{-}$ and $v_{+}$, respectively,  the minimum and maximum of 
$\frac{1}{\al} (\frac{\vert F_1 \vert ^2}{2} - R - T_{string}) $.
If  $\be$ is sufficiently close to $\frac{-2\al}{3}$
and $F_p$ is close enough to $0$ in the $H^s$-topology, except 
possibly for $p=1$, and if $\frac{v_{+}}{v_{-}} < 2$, then there exists a unique $H^s$ solution $(v,u)$ to the system (\ref{eq_motion}) in a small
$H^s$-neighbourhood of $(v_{*},1)$, where $v_*$ is a solution of (\ref{eq_v_sub_sup})
satisfying $v_- \leq v_* \leq v_+$.  The solution satisfies $v > 0$, $u > 0$, and depends continuously on $\al$, $\be$, $F_p$, $R$, and $T_{string}$.
\label{theo_sol_unstable}
\end{theorem}
\begin{rema}
By taking $s$ sufficiently large and applying the Sobolev embedding theorem,
we can replace the $H^s$-neighborhood with a $C^k$-neighborhood in the previous statement. A similar statement holds for the other theorems presented below.
\end{rema}
We notice that $v>0$ and $u>0$ follow by making the $H^s$-neighborhood 
of the theorem very small and using $v_* > 0$; 
since $s$ is assumed to be large, these solutions
are in fact continuous and the pointwise inequalities $v>0$ and $u>0$ then hold.
We also remark that starting with smooth $F_1$, $R$, and $T_{string}$
does not necessarily yield smooth solutions. This is because the perturbed
$F_p$ produced by the implicit function theorem will, in general, be only
in $H^s$. If they happen to be smooth, however, then $v$ and $u$ are smooth 
due to elliptic regularity. Indeed, if
$M(v,u) = 0 = N(v,u)$, then $vN(v,u) - u M(v,u) = 0$, which takes the form
\begin{gather}
4 u v^2 \Delta u = f(u,v,\nabla u, \nabla v),
\nonumber
\end{gather}
where $f$ is a smooth function of its arguments.
Since $v,u \in H^s$, $f(u,v,\nabla u, \nabla v) \in H^{s-1}$ and so
$u \in H^{s+1}$ by elliptic regularity. $M(v,u)=0$ and elliptic regularity then
give $v \in H^{s+1}$, and bootstraping this argument, we conclude that $v,u$
are smooth.

We now show that solutions given by theorem \ref{theo_sol_unstable} are in general unstable. Evaluating (\ref{mass_general}) at $(v_*,1)$ yields
\begin{gather}
\left. \frac{ \partial^2 V_{eff} }{\partial v^2} \right|_{(v_*,1)} = \frac{1}{2}
\int_M ( h - 2\al v_* ),
\nonumber
\end{gather}
where $h$ is as above. Multiplying (\ref{eq_v_sub_sup}) by $v_*$ and integrating by parts,
\begin{gather}
\int_M ( h - 2\al v_* ) = - \int_M \frac{1}{v_*^2} |\nabla v_*|^2 \leq 0,
\nonumber
\end{gather}
so that $\left. \frac{ \partial^2 V_{eff} }{\partial v^2} \right|_{(v_*,1)}  \leq 0$, 
and the inequality is in fact strict if $v_*$ is not constant. Notice that
$v_* = constant$ will not be a a solution of (\ref{eq_v_sub_sup}) unless
further relations among the scalar curvature, the gauge fields, and the string term hold.
It follows that the strict inequality will be preserved for  solutions
very close to $(v_*,1)$, and we conclude:

\begin{coro}
If $v_*$ is not constant, then the solutions $(v,u)$ given by theorem
\ref{theo_sol_unstable} are unstable in the sense of definition
\ref{defI_modulus}, provided that $(v,u)$ is sufficiently close to 
$(v_*,1)$.
\end{coro}

\subsection{Stable solutions and applications to Type IIB strings}
Using different hypotheses than those of the previous section, here 
we prove existence of stable solutions to (\ref{eq_motion}). The key
ingredient is to balance the contribution of the gauge fields with that
of the source $T_{string}$. This is, in fact, an idea that goes back to 
\cite{GKP, KKLT} and has been extensively used in moduli stabilization
\cite{Douglas:2006ab, Grana:2006ab}. 
We shall impose conditions
that lead to a direct application to dS-vacuum in Type IIB strings. In fact, 
we shall provide a set of slightly different theorems applicable to different
Type IIB scenarios\footnote{
It will be clear from what follows that similar arguments can  be constructed in 
other settings.}. For the rest of this section we suppose the following.

\begin{assump}
Let $q=3$, and we suppose from now on that
$F_0 = F_2 = F_4 = F_6 = 0$, as these fluxes are not present 
in Type IIB compactifications. 
\end{assump}

The arguments that follow will be similar 
to the proof of theorem \ref{theo_sol_unstable}, and they are all more or less
analogous to each other. Hence, in order to avoid being repetitive,  we shall 
go through them rather quickly.  
Set $q=3$, $\al=\be=R=0$ and $F_p = 0$ for $p \neq 3$ in (\ref{eq_motion}),
and plug in $v = u = 1$. Then (\ref{eq_motion_2}) holds identically and
(\ref{eq_motion_1}) is satisfied provided that
\begin{gather}
-2 T_{string} + |F_3|^2 = 0,
\label{balance}
\end{gather}
so we hereafter assume this condition. We point out that, other
than being similar to previous balance conditions used in 
uplifting \cite{Grana:2006ab, GKP, Douglas:2006ab, KKLT}, 
(\ref{balance}) also resembles a local version of the tadpole 
cancellation appearing in \cite{GKP, KST}.
In fact, in supersymmetric solutions, like in \cite{KST}, one can see 
that the contribution of the 3-flux to the overall potential gets a 
cancellation from localized 
sources. This follows as an application of the integrated Bianchi identity. 
In our case, the assumption shows the cancellation locally. Thus,
it would be  interesting to explore the relationship between this assumption and the Bianchi identity.

\begin{theorem}
Let $q=3$, assume that $|F_3|^2-2 T_{string}= 0$,
and fix a 
sufficiently large real number $s$. If $\al$, $\be$ are close enough to zero, and $R$,
$F_p$ are sufficiently close to zero in the $H^s$ topology, 
except 
possibly for $p=3$,
then the system (\ref{eq_motion}) 
has a $H^s$ solution $(v,u)$ in a small $H^s$-neighborhood of $(1,1)$.
Such a solution satisfies $v>0$, $u>0$,
and it depends continuously on $\al$, $\be$, $R$, $F_p$, $T_{string}$,
 and on two real parameters, 
$l_1,l_2$ (these parametrize the kernel of (\ref{linear_Laplacian}) below).
\label{theo_sol_stable}
\end{theorem}
\begin{proof}
Let
\begin{gather}
W = 10 uv \Delta v + 6v^2 \Delta u + 12 v \nabla v \cdot \nabla u
-2uv^2R + u \sum_{p=1,3,5} v^{(3-p)} |F_p|^2 - 2 u T_{string} - 2 \al v^3,
\nonumber
\end{gather}
and 
\begin{gather}
Z = 10 u^2 \Delta v + 10 uv \Delta u + 20 u \nabla u \cdot \nabla v + 10
uv |\nabla u|^2 + 4v |\nabla u |^2 -2u^2 v R 
\nonumber \\
 + \frac{u^2}{2} \sum_{p=1,3,5} (3-p)v^{2-p} |F_p|^2 
 - 6\al u v^2 - 6\be v^2.
\nonumber
\end{gather}
Consider the map
\begin{gather}
J: \RR \times \RR \times H^s \times \dots \times H^s \times H^s \mapsto
H^{s-2} \times H^{s-2}, \nonumber \\
(\al,\be, |F_1|^2, |F_3|^2, |F_5|^2, R, T_{string} ,v, u ) \mapsto (W,Z).
\nonumber
\end{gather}
Solutions are given by $W = 0 = Z$, and $u=v=1$ is a solution 
when $\al=\be=F_1=F_5=R = 0$ and (\ref{balance}) holds. Linearizing
at $(1,1)$ and in the direction of $(\chi_1, \chi_2)$, and setting it equal 
to $(\psi_1,\psi_2)$ gives
\begin{gather}
\begin{cases}
\Delta \chi_1 = \psi_1, \\
\Delta \chi_2 = \psi_2.
\end{cases}
\label{linear_Laplacian}
\end{gather}
As $M$ is compact without boundary, its harmonic functions are constant and
(\ref{linear_Laplacian}) has a unique solution modulo additive constants. 
The implicit function theorem is not, therefore, directly applicable, but we can still
rely on it after restricting the linearization to the $L^2$-orthogonal to its kernel, 
what produces solutions near $(1,1)$. A parametrization of the kernel
of (\ref{linear_Laplacian}) yields the parameters $(l_1,l_2)$ of the theorem.
As in the previous section, 
these solutions are positive if the $H^s$-neighborhood is sufficiently small.
\end{proof}

Next, we investigate stability.  With 
$q=3$, $\al=\be=R=0$, and $F_p = 0$ for $p \neq 3$, evaluating (\ref{mass_general}) 
at $(1,1)$ yields
\begin{gather}
\left. \frac{ \partial^2 V_{eff} }{\partial v^2} \right|_{(1,1)} = 0.
\label{deriv_V_zero}
\end{gather}
The ``uplifting" approach \cite{GKP, KKLT}
to constructing positive energy vacua 
consists, in a nutshell, of starting with an AdS supersymmetric vacuum ($\al < 0$),
where stability can be achieved, and deforming the data to a dS ($\al > 0$) vacuum.
With proper control of this deformation, the new vacuum can be shown
to be stable. Furthermore, experience shows that $\al < 0$  
generally favors stability
\cite{Disconzi:2012a, Douglas:2006ab}.
Therefore, on physical grounds, we expect the continuous dependence on the data
guaranteed by theorem \ref{theo_sol_stable} to allow us to continue solutions
from $\al=0$ to $\al > 0$, while changing the equality on
(\ref{deriv_V_zero}) to a strict ``greater than" inequality.  
The favorable physical arguments and apparent absence of a preventative  
mechanism notwithstanding, in order to prove stability, further hypotheses are needed.
Interestingly enough, such hypotheses concern the value of some
natural constants that arise in elliptic theory and which are ultimately 
tied to the topology and geometry of $M$. This is consistent with 
experience in compactifications, where global properties of the compact
manifold play an important role in moduli stabilization.

\begin{theorem}
(dS stability in Type IIB with slowly varying warp factor)
Assume the same hypotheses of theorem \ref{theo_sol_stable}.
Let $K_1$ be the norm of the map $\Delta:  H^{s}_0 \rar H^{s-2}$, where
$H_0^{s} := H^s/\operatorname{ker}\Delta$, and let $K_2$ be the best Sobolev
constant of the embedding $H^s \hookrightarrow C^1$.
If $\frac{K_2}{K_1} < 26$ holds\footnote{This will be the case, for example,
if $g_{ij}$ is sufficiently close to the Euclidean metric.}, then it is possible to choose $\al >0$, $\be<0$, 
$F_1$, $F_5$, and $R$ all sufficiently small, 
such that the corresponding solutions $(v,u)$
with $l_1=l_2 = 0$
given by theorem \ref{theo_sol_stable} are stable 
in the sense of definition \ref{defI_modulus}.
\label{dS_stable}
\end{theorem}
\begin{proof}
Compute
\begin{gather}
D_x J (x_0,1,1)
 \chi
= 
\left(
\begin{array}{ccccccc}
-2 & 0 & 1 & 1 & 1& - 2 & - 2 \\
-6 & - 6 & 1 & 0 & -1 & -2 & 0
\end{array}
\right) \cdot \chi,
\nonumber
\end{gather}
where $x$ is shorthand for $(\al,\be, |F_1|^2,|F_3|^2, |F_5|^5, R, T_{string})$,
$x_0$ stands for $(0,0,0,2T_{string}, 0, 0, T_{string})$, 
and $\chi = (\chi_1, \dots, \chi_7) \in \RR \times \RR \times H^s \times\dots
\times H^s$. From this expression it follows that 
\begin{gather}
\p D_x J \p < 26,
\nonumber
\end{gather}
if $(x,v,u)$ is sufficiently close to $(x_0,1,1)$.
By theorem \ref{theo_sol_stable}, we have solutions
$J(x,v(x),u(x) ) = 0$, and the implicit function theorem 
guarantees that the map $x \mapsto (v(x),u(x))$ is differentiable.
 The map $D_{v,u} J$, being an isomorphism
at $(x_0,1,1)$, will be invertible nearby, thus
\begin{gather}
D_x 
\left( 
\begin{array}{c}
v \\
u 
\end{array}
\right) 
= - \left( D_{v,u} J\right)^{-1} D_x J.
\nonumber
\end{gather}
Keeping 
$\al=F_1=F_5=R=0$ and $|F_3|^3 = 2T_{string}$,
we find, under the assumptions of the theorem and invoking
the Sobolev embedding theorem, that
\begin{gather}
\p u- 1 \p_{C^1} \leq |\be|,
\nonumber
\end{gather}
where $\p \cdot \p_{C^1}$ is the standard $C^1$ norm.
Furthermore, since $u = v \equiv 1$ when $\be = 0$, we can choose $\be < 0$
so small that 
\begin{gather}
\frac{1}{2}\int_{M} -6 |\nabla u |^{2}
- 6 \beta\int_{M} v  > 0.
\nonumber
\end{gather}
From (\ref{mass_general}), it now follows that we can pick
$\al > 0$ and the remaining fields so small that
\begin{gather}
\frac{\partial^2 V_{eff}}{\partial v^2} > 0,
\nonumber
\end{gather}
finishing the proof.
\end{proof}
From (\ref{mass_general}), we see that $\beta < 0$ favors stability, hence
it would be natural to expect that $\frac{\partial^2 V_{eff}}{\partial v^2} > 0$
if $\beta \ll -1$. Theorem \ref{theo_sol_stable}, however, only guarantees the
existence of solutions for $\beta$ close to zero. We therefore consider another
condition which allows $\beta$ to be considerably negative, namely,
\begin{gather}
|F_3|^2  - 2T_{string} = 2R-|F_1|^2 = -6 \beta.
\label{balance_beta_negative}
\end{gather}
We readily check that if (\ref{balance_beta_negative}) holds and
$\al =F_5 = 0$, then $v = u = 1$ is a solution of (\ref{eq_motion}).
Arguing similarly to the proof of theorem
\ref{theo_sol_stable} and recalling (\ref{mass_general}) leads to:

\begin{theorem}
(dS stability in Type IIB with slowly varying warp factor and $\beta \ll -1$)
Let $q=3$, assume (\ref{balance_beta_negative}),
and fix a 
sufficiently large real number $s$. If $\al$ is close enough to zero,  
$F_5$ is sufficiently close to zero in the $H^s$ topology, 
and $\beta$ is sufficiently negative, 
then the system (\ref{eq_motion}) 
 has a $H^s$ solution $(v,u)$ in a small $H^s$-neighborhood of $(1,1)$.
 Such a solution satisfies $v>0$, $u>0$,
depends continuously on $\al$, $\be$, $R$, $F_p$, $T_{string}$,
 and on two real numbers $l_1,l_2$
 parametrizing the kernel of (\ref{linear_Laplacian}). Furthermore,
 this solution is stable in the sense of definition \ref{defI_modulus}, provided that 
we choose $l_1=l_2 = 0$.
\label{dS_stabe_beta_negative}
\end{theorem}

\begin{rema}
Although upon setting $\beta = R = F_1 = 0$, (\ref{balance}) can be obtained from (\ref{balance_beta_negative}), the interest in the latter is, of course, when $\beta$
is large negative, as in theorem \ref{dS_stabe_beta_negative}, in which case,
theorem \ref{dS_stable} does not apply.
\end{rema}

A clear limitation of theorem \ref{dS_stabe_beta_negative} is the lack
of a precise bound on how negative $\be$ ought to be. This is important 
because, since $\be$ is related to $\Vol(M)$, a large $|\beta|$
might be out of the supergravity limit in parameter space. We still find it interesting
to state theorem \ref{dS_stabe_beta_negative}, however, because a closer inspection
suggests that a moderately 
negative $\beta$ should suffice, as long as something like (\ref{balance_beta_negative})
holds. Confirming this requires a sharper understanding of the solutions
to (\ref{eq_motion}), which will be carried out elsewhere \cite{in_prep}.
Furthermore, we can still study compactifications in the low energy limit 
without imposing supersymmetry, and hence considering
supergravity, provided that we solve the full higher
dimensional Einstein's equations coupled to matter. We could then start with solutions 
with $\be \ll -1$, and by varying $\be$ towards zero, investigate how far in parameter 
space the low energy description remains valid.

Together, theorems \ref{theo_sol_stable}, \ref{dS_stable},
and \ref{dS_stabe_beta_negative} establish satisfactory properties
of the effective description with a slowly varying warp factor;
they give conditions for existence of solutions to the equations of motion
and their stability. Many times, however, it is important
to have a more explicit account of the functions $u$ and $v$.
In particular, one is interested in expanding $u = 1 + \ve$,
where $\ve$ is a fairly tractable (perhaps explicit) small function.
To accommodate this situation, we
turn our attention to another existence
result. 

Set $\cF_p = |F_p|^2$, $a=2R-\cF_1$, and $b=\cF_3-2T_{string}$. Let 
\begin{gather}
X = 10 uv \Delta v + 6v^2 \Delta u + 12 v \nabla v \cdot \nabla u
-a u v^2  + u  v^{-2} \cF_5
+ b u  - 2\al v^3,
\nonumber
\end{gather}
and
\begin{gather}
Y = 10 u^2 \Delta v + 10 uv \Delta u + 20 u \nabla u \cdot \nabla v + 10
uv |\nabla u|^2 -a u^2 v  + 4v |\nabla u |^2
 - u^2 v^{-3} \cF_5 - 6\al u v^2 - 6\be v^2,
\nonumber
\end{gather}
so that solutions to (\ref{eq_motion}) are given by $X=0 = Y$. Notice that $u=v=1$ is a solution to $X=Y=0$ if $\cF_5=0$, $\alpha=0$, and $a=b=-6\beta >0$. \\
\indent It is easily seen that if $X=Y=0$, then by inverting a matrix one can solve for $a$ and $b$ in terms of $(\al, \be, \cF_5, u, v, \nabla u, \nabla v, \Delta u, \Delta v)$. In other words,

\begin{lemma}
Let $q=3$, assume that $F_0=F_2=F_4=F_6 \equiv 0$,
and fix a 
sufficiently large real number $s$. If $(v,u)$  is sufficiently close to $(1,1)$
in the $H^s$ topology, then the following holds.
There exist $H^{s-2}$ functions $a$ and $b$, depending continuously
on $\cS:=(\al$, $\be$, $F_5$, $v$, $u$) such that
the system (\ref{eq_motion}) is satisfied upon
replacing $2R-|F_1|^2$ and $|F_3|^2-2T_{string}$ with $a$ and $b$, respectively, 
and the remaining data taking the values given by $\cS$.
\label{theo_sol_stable_2}
\end{lemma}

We comment on the difference between theorem \ref{theo_sol_stable}
and lemma \ref{theo_sol_stable_2}, which consists basically in  
what is treated as independent or dependent data. Theorem 
\ref{theo_sol_stable} states that if we start with the solution $(1,1)$ 
and slightly perturb the other fields,
 then it is possible to find functions near $(1,1)$ that satisfy the 
perturbed equation. Lemma \ref{theo_sol_stable_2}, on the other hand, 
says that if we choose any two functions
$u$ and $v$ close to one, then we can find functions fitting the remaining 
data of the equation (namely, $a$ and $b$) in order to force 
$v$ and $u$ to be solutions. Although theorem \ref{theo_sol_stable}
is a more standard existence theorem, lemma \ref{theo_sol_stable_2} 
has the advantage of allowing one to explicitly construct solutions
of the form $u=1 + \ve_1$ and $v = 1 + \ve_2$, which are useful in 
asymptotic analysis of the problem. 

\begin{coro}
It is possible to choose $\al >0$, $\be<0$, $F_5$,
all sufficiently small, and $(v,u)$ sufficiently close to $(1,1)$, such that
the corresponding 
solutions given by lemma \ref{theo_sol_stable_2} are stable 
in the sense of definition \ref{defI_modulus}.
\end{coro}
\begin{proof}
This corollary is a direct consequence of the following identity, which is, in turn, proven by dividing equation (\ref{eq_motion_2}) by $v$ and integrating,
\begin{gather}
\frac{\partial ^2 V_{eff}}{\partial v^2} = -10 \int_M \frac{u^2 \vert \nabla v \vert^2}{v^2} + \int_M u^2 R + 3 \int_M \vert \nabla u \vert^2  - \frac{1}{2} \int_M u^2 \cF_1 + \frac{5}{2} \int_M u^2 v^{-4} \cF_5.
\nonumber
\end{gather}
\end{proof}

We finally comment on the fact that, in order to construct solutions,
we allowed the several fields in (\ref{eq_motion}) to vary. In other words,
we are not solving for $v$ and $u$ given fixed data $\al$, $\be$,
$R$, $F_p$ and $T_{string}$, but these data themselves are allowed
to change slightly so that solutions are found. Had we been investigating
a set of equations taken as the fundamental equations of a theory,
this approach would certainly be problematic. The warped-conformal factor
system (\ref{eq_motion}), however, is only an approximation to the fundamental
set of equations of string/M-theory, and as such, they can be slightly adjusted
as long as their relevant physical content is kept unchanged. 
This possibility of tweaking  the several fields involved is, in fact,
what has allowed physicists to pursue programs like moduli stabilization
and the construction of meta-stable positive energy vacua.

\section{Volume estimates and non-perturbative effects}

In this section we derive some basic identities and inequalities that relate 
$\Vol(M)$ with the other quantities of the problem. 
We assume throughout this sections that we are 
given  positive solutions $u$ and $v$ of  (\ref{eq_motion}). Because
$n=6$, we must have $L \leq 6$. Notice also
that many of the bounds below involve integral quantities,
and, therefore, the smoothness of Assumption \ref{assumption_smooth} 
can be relaxed, as mentioned in Remark \ref{remark_smooth}.

\begin{lemma} The following identity holds:
\begin{gather}
\displaystyle \int _M (\frac{u^2}{2} \sum _{p=0}^L  (1-p) v^{3-p} \vert F _p \vert ^2 + \frac{7-q}{2}u^2 v^{(q-3)/2}T_{string} ) = \frac{4\al}{G_N}+6\be \Vol(M).
\label{identity_al_be_vol}
\end{gather}
\label{lemma_al_be_vol}
 \end{lemma}
 \begin{proof}
 Equations (\ref{eq_motion}) can be written as
\begin{gather}
12u\nabla v \cdot  \nabla u + 6uv\Delta u + 10 u^2 \Delta v - 2u^2 vR + u^2 \sum _{p=0}^{L} v^{2-p} \vert F_p \vert^2 -2u^2 v^{(q-5)/2}T_{string}  = 2\alpha v^2 u , \nonumber 
\end{gather}
and
\begin{gather}
20 u\nabla v \cdot \nabla u + 10uv \Delta u + 10 u^2 \Delta v -2u^2 vR + 4 \vert \nabla u \vert ^2 v  \nonumber \\
+
  \frac{u^2}{2} \sum _{p=0}^L  (3-p) v^{2-p} \vert F _p \vert ^2- \frac{q-3}{2}u^2 v^{(q-5)/2}T_{string}  = 6 (\alpha v^2 u + \beta v^2).
\nonumber
\end{gather}
Subtracting the first from the second, 
\begin{align}
\begin{split}
& 
8 u\nabla v \cdot \nabla u + 4uv \Delta u + 4 \vert \nabla u \vert ^2 v 
+ \frac{u^2}{2} \sum _{p=0}^L  (1-p) v^{2-p} \vert F _p \vert ^2 \\
&
+ \frac{7-q}{2}u^2 v^{(q-5)/2}T_{string} = 4 \alpha v^2 u + 6\beta v^2. 
\end{split}
\label{usefulequation}
\end{align}
Multiplying equation (\ref{usefulequation}) by $v$, integrating and
using (\ref{constraints}) gives (\ref{identity_al_be_vol}).
\end{proof}

Lemma \ref{lemma_al_be_vol} can be used to derive useful relations
among $\al$, $\be$, and the volume of compact dimensions. The only
positive contribution from the gauge fields to the right hand side of 
(\ref{identity_al_be_vol}) comes from $F_0$, hence one expects that 
it is possible 
to tune the gauge fields' contribution in order 
to balance that of the string term so that
\begin{gather}
\frac{4\al}{G_N}+6\be \Vol(M)\leq 0 .
\label{inequality_al_be_vol}
\end{gather}
It follows that if $\al >0$ (which is the main case of interest, 
as previously pointed out), then necessarily $\be < 0$. This is important
because, in general, no simple physical condition fixes $\be$, and in this case
we also have the bound
\begin{gather}
 \Vol(M) \geq \frac{2\al}{3 G_N |\beta| }.
\nonumber
\end{gather}

The standard strategy in string compactification is to perform the analysis in the supergravity limit and consider string theory as the UV cutoff for the effective field theory. This imposes the radius $R_c$, and hence the volume of $M$, of the compactified six dimensions to be much larger than the string length,  $\ell_{string} = \sqrt{\al'}$, and moderately weak coupling, $g_{s}\rightarrow 0$. The effective
$4$-dimensional Planck scale, $M_{Pl,4}^2 = \frac{1}{G_{N}}$, is determined by the fundamental 10-dimensional Planck scale, $M_{Pl,10}$ (set to be equal to one in \cite{Douglas:2009a}), and the geometry of the extra dimensions with warping. Thus, $\frac{1}{G_N}=\Vol(M)_{warped}$. We have no experimental signs of the extra dimensions because the compactification scale, $M_c \sim 1/(\Vol(M)_{warped})^{1/n}$ would have to be smaller than the observable particle physics scale. Thus, in general, we would like to set $\be$ such that the supergravity limit is valid (i.e. $\Vol(M) >>\ell_{string}^6$), and the compactification scale is beyond the current observable scale (which is TeV scale in standard units). A detailed physics discussion can be found in \cite{Douglas:2006ab}.

As stated, Lemma \ref{lemma_al_be_vol} is quite general, and hence
inequality (\ref{inequality_al_be_vol}) should be valid under 
a wide range of scenarios. Given a particular model, however, it
may be difficult to verify that the integral on the 
left hand side of (\ref{identity_al_be_vol}) is negative. We therefore
point out two further situations involving point-wise rather than integral conditions, where (\ref{inequality_al_be_vol}) holds,
and in which the verification of the hypotheses is more direct.
One is when  $q=3$, $F_0 =0$,  and $2T_{string} -\vert F_3 \vert ^2 \leq 0$.
The other is when $q=7$ and $F_0 = 0$. In both cases, it follows at 
once that the left hand side of (\ref{identity_al_be_vol}) is non-positive.

Any geometric compactification of string/M-theory has a “large volume limit” which 
approaches ten or eleven dimensional Minkowski space. In this limit, the four 
dimensional effective potential vanishes. De Sitter models from string compactifications 
are difficult to construct because, as non-supersymmetric vacua, they are isolated points 
in the moduli space with all moduli stabilized. To understand de Sitter solutions, one 
must have sufficient understanding of non-perturbative effects to show that such 
potentials could come from string theory, and, moreover, could be computed in some 
examples in order to make real predictions.  Such effects usually give AdS vacua with all 
moduli stabilized, and then uplifting is achieved with a proper classical contribution.
This has been an active research topic within string compactifications.

Since we are looking at Type IIB theory mainly, non-perturbative effects are added to stabilize the K\"ahler moduli. In our approach, we have not included such effects yet. Non-perturbative effects give a contribution to the overall potential in the following way \cite{Giddings}:
\begin{gather}
V_{non-pert} = Be^{−2av^{4}}/v^{s},
\nonumber
\end{gather} such that the constants $a$ and $s$ arise from euclidean D3 brane instantons or gluino condensation, as explained in \cite{KKLT}. We want to study the conformal factor classically, and we claim that non-pertubative effects are added such that they do not affect the critical point and its mass noticeably. This can be written qualitatively as
\begin{gather}
\frac{\partial^{2}V_{eff}}{\partial v^2} >> \frac{\partial^{2}V_{non-pert}}
{\partial v^2}
\nonumber
\end{gather}
at critical points.
Thus, we should expect to find locally stable minima with positive cosmological constant.

\end{document}